%% file: main.tex
\documentclass[]{IEEEtran}
\usepackage{nccfoots}
\usepackage{balance}
\usepackage[usenames,dvipsnames]{color}
\usepackage[tight,TABTOPCAP]{subfigure}
\usepackage{caption2}
\usepackage[cmex10]{amsmath}
\usepackage{multirow}
\usepackage{algorithmic}
\usepackage{amssymb}
\usepackage{color}
\usepackage{soul}

\DeclareMathOperator*{\argmin}{arg\,min}

\input{new-commands}
\usepackage{amsfonts}
\usepackage{graphicx}

\newtheorem{definition}{Definition}

\begin{document}

\title{Understanding Relative Network Delay in Micro-Energy Harvesting Wireless Networks}

\author{
  \IEEEauthorblockN{
    Qi Chen\IEEEauthorrefmark{1} and
    Qing Yang\IEEEauthorrefmark{2}
  }
  
  \IEEEauthorblockA{
    \textit{Department of Computer Science and Engineering} \\
    \textit{University of North Texas, USA}
  }
  
  \IEEEauthorblockA{
    \IEEEauthorrefmark{1}\{QiChen\}@my.unt.edu,
    \IEEEauthorrefmark{2}\{Qing.Yang\}@unt.edu
  }
}

\maketitle

\begin{abstract}
Micro-energy harvesting wireless network (MEHWN) enables a perpetual network deployment that cannot be achieved in traditional battery-operated counterparts. Despite its sustainability, end-to-end delay in an MEHWN could be very large, due to the large waiting delay on each hop in the network. In this work, we consider an MEHWN where every node constantly switches between on and off states, due to the limited amount of harvested energy. The network delay of an MEHWN is not well understood because of the energy uncertainty, asynchronized working schedules, and complex network topology in an MEHWN. To close this research gap, we define the relative network delay as the ratio between end-to-end delay and distance. Compared to previous works, we are able to identify a closed-form expression of the lower bound and a tighter upper bound of the relative network delay. The theoretical findings are verified in simulations. Our theoretical analysis deepens the understanding about the interplay of network delay, energy harvesting rate, and node density in an MEHWN.
\end{abstract}


{\sloppy
\section{Introduction}
Micro-scale energy-harvesting technologies that scavenge milliwatts/microwatts from ambient sources can provide power to wireless devices/sensors to form a micro-energy harvesting wireless network (MEHWN)~\cite{yang2017, xue2017water, survey}.
For common commercial of the shelf (COTS) wireless devices, their energy consumption is usually hundreds of milliwatts per second. 
As a result, the energy harvested in an MEHWN is not enough to continuously provide power to COTS wireless devices.
Consequently, a device in the MEHWN will periodically switch between on and off: when the device is off, energy is harvested and  stored in either a battery or super capacitors; after enough energy is collected, the device will be turned on. 
This type of network uses virtually inexhaustible energy sources and has little or no adverse environmental effects.
Potential applications of an MEHWN span from structural health monitoring~\cite{structural}, sustainable sensor networks~\cite{sensornet}, to cognitive networks~\cite{cognet} and the Internet of Things (IoT)~\cite{IoT}.

\subsection{Problem Statement}
MEHWN is a new type of wireless network that is different from existing ones, which brings more research challenges. 
For example, the end-to-end delay in an MEHWN will be significantly different from that in existing wireless networks because the time needed for nodes to harvest enough energy will dominate the network delay. 
In an MEHWN, we know the expected waiting delay on each hop is very large, e.g., minutes to hours~\cite{yang2017}.
To better understand the end-to-end delay in an MEHWN, it is valuable to investigate the relative network delay, defined as the delay-to-distance ratio.
Thanks to the percolation theory~\cite{meester1996continuum} and first passage time~\cite{firstpassage}, it is possible to study the theoretical upper and lower bounds of an MEHWN's relative network delay.
Theoretical findings about relative network delay will provide guidelines in designing and deploying an MEHWN.

\subsection{Limitations of Prior Art}
There exist several works on energy harvesting based wireless networks~\cite{ yang2012optimal,xu2014throughput,ng2013energy}.
Most existing works, however, mainly focus on single-hop energy harvesting wireless networks, i.e., there are fewer works on multi-hop MEHWNs~\cite{Zhu2013}. 
Recently, Shizhen \textit{et al}. proposed the seminal work of using percolation theory to study the fundamental relationship between node density and delay in a wireless ad hoc network with unreliable links~\cite{Wang11}. 
To find the lower bound of relative network delay, they discovered there exists a set of clusters in the network from the source to destination. 
A cluster consists of a group of nodes that are instantaneously connected to each other via multi-hop wireless communications.
The size of a cluster is defined as the largest distance between two nodes in it.
The cluster to cluster transmission is considered a series of outbursts. 
During each outburst, several new nodes are connected, and finally the destination is reached. 
Based on this concept, an approximated lower bound of relative network delay is found in~\cite{Wang11}.
However, a closed-form expression of the expected cluster size is not provided.
We close this gap by identifying an upper bound of the expected cluster size, and provide a closed-form expression of the lower bound of relative network delay.

\subsection{Technical Challenges and Proposed Solutions}
The first technical challenge is computing the upper bound of the expected cluster size.
It is a challenge because it is difficult, if not impossible, to know which node is included in which cluster.
We address this challenge by coupling an MEHWN with a square lattice. 
For each cluster in the MEHWN, we identify a corresponding connected component that consists of a set of connected edges in the lattice. 
As such, we are able to use a connected component to approximate a cluster.

The second technical challenge is obtaining the upper bound of a connected component's diameter.
The diameter of a connected component is measured by the number of edges along the horizontal direction.  
This is a challenge because a connected component could be in various forms, e.g., linear, grid, or most likely irregular shapes. 
To address this challenge, we first compute the probability that an $n$-size connected component (containing $n$ vertices) has the diameter of $k$.
Then, the expected diameter is computed by considering all possible $k$'s.

The third technical challenge is identifying a closed-form expression of the lower bound of relative network delay.
This is a challenge because a connected component's diameter measures the horizontal distance of the underlying cluster, not the actual cluster's size.
We address this challenge by rotating the original network, \ie, we treat the line from the source to the destination as the new horizontal axis.
In this way, we can use a connected component's diameter to approximate a cluster's size. 
Finally, we are able to identify a closed-form expression of the relative network delay's lower bound.

\section{Background and System Models}
\label{sec:background}

\subsection{Random Connection Model}
Our study considers the scenarios where long-term connectivity exists in an MEHWN, \ie, a packet could eventually be delivered from a node to another in the network. 
We make this assumption because an MEHWN would be useless if long-term connectivity does not exist.
To model the node deployment in an MEHWN, we adopt the random connection model (RCM) that is commonly used to study network connectivity in large-scale networks.
In the RCM, denoted by $G(\lambda, r_0, g)$, nodes are distributed according to the Poisson Point Process~\cite{ppp} with the density $\lambda$. 
The communication range between two nodes is $r_0$. 
A communication link exists between two nodes with a probability of $g$.

Based on the percolation theory~\cite{meester1996continuum}, there exists a critical node density $\lambda_c$, such that when $\lambda \geq \lambda_c$, there is a giant cluster containing almost all nodes in the network. 
The giant cluster, denoted by $\mathcal{C}(G(\lambda, r_0, g))$, is called the \textit{giant component}~\cite{meester1996continuum}. 
In this case, the network is called \textit{percolated}. 
If $\lambda < \lambda_c$, there is no giant component but several smaller clusters in the network, and each of them contains finite number of nodes.

Using the RCM, we consider an MEHWN can be percolated in two different ways.
With a high node density, an MEHWN could be connected at any time, and we call it instantaneously percolated.
In this case, its node density must be greater than or equal to the \textit{instantaneous critical density}, denoted by $\lambda_I$. 
At a certain time $t$, let $G_t(\lambda, r_0, g)$ denote the instantaneous graph of $G(\lambda, r_0, g)$, we define $\lambda_I$ as follows.
\begin{definition}
Instantaneous critical density $\lambda_I$ is
\begin{equation}
\lambda_I = \inf \{ \lambda |\text{$\forall t$, $G_t(\lambda, r_0, g)$ is percolated}\}.
\end{equation}
\end{definition}

When the node density is low, an MEHWN can still be percolated in the long run. 
Here, we define the \textit{long-term critical density}, denoted by $\lambda_L$, as follows.

\begin{definition}
Long-term critical density $\lambda_L$ is
\begin{equation}
 \lambda_L = \inf \{ \lambda |\text{$ G'(\lambda, r_0, g)$ is percolated}\},
\end{equation}
where $\displaystyle G'(\lambda, r_0, g) = \cup_{t}^{\infty} G_t(\lambda, r_0, g)$ contains the nodes and edges from all $G_t(\lambda, r_0, g)$.
\end{definition}

\subsection{Energy Harvesting Model}
It is difficult and complex to accurately model the energy harvesting process on a node in an MEHWN.
The amount of energy available on a node depends on not only the harvesting technique, energy storage methods, and the type of ambient sources, but also the network traffic generated and forwarded by this node.
To enable the analysis of relative network delay in an MEHWN, we adopt the binary energy harvesting model where a node is considered to have a unit energy storage capacity, i.e., the energy harvested on a node may overflow and get wasted~\cite{tutuncuoglu2014state,tutuncuoglu2014improved,ozel2015binary}.
The energy availability of a node is then modeled as an ON/OFF process $\{ X_i(t), t \geq 0\}$ 
\begin{equation*}
X_i(t) = \left\{
\displaystyle
  \begin{array}{l l}
    1 & \quad \text{if node $i$ is available at time $t$,}\\
    0 & \quad \text{otherwise.}
  \end{array} \right.
\end{equation*}

We assume $\{X_i(t)\}$ is an alternative renewal process, and define the probability of node $i$ being active as $q_i$. 
We assume the value of $q_i$ could be obtained from analyzing off-line/historical data~\cite{fan2008steady}. 
Therefore, two node $i$ and $j$ are connected with the probability of $g_{ij} = q_i q_j $, if the distance between nodes $i$ and $j$ is smaller than the communication range.
Here we assume $q_i$ is the same for every node, denoted by $q$, then the connecting probability between any two nodes can be represented as $q^2$.

\subsection{Network Model}
Now we will introduce the network model where the networking time is divided into slots.
Due to slow energy harvesting rates (mW/s or $\mu$W/s), we assume a few seconds/minutes waiting delay exists on each wireless link.
Compared to the delays in wireless signal propagation, information processing, and transmission scheduling that are usually less than a second, the waiting delay on each link will dominate the end-to-end delay in an MEHWN.
In this work, we ignore others but focus on the waiting delay on each link in the network. 

If a sender fails to deliver a packet in a time slot (with the probability of $1-g$), we assume it tries to re-transmit this packet in the next time slot.
So the waiting delay on this link can be modeled as
\[
\Pr\{T(e) = z\} = (1-g)^z g,
\]
where $T(e)$ denotes the waiting delay on the link $e$ and $z$ denotes the number of time slots.
Therefore, the expected waiting delay is
\[
E(T(e)) = \frac{1}{g}-1
\]

Given $E(T(e))$, we adopt the first passage time~\cite{firstpassage} to model the end-to-end delay in an MEHWN. 
In the RCM $G(\lambda, r_0, g)$, the first passage time from nodes $u$ to $v$ is defined as
\begin{equation}
T_\lambda(u, v) = \inf \{ T(\pi)|\text{$\pi$ is a pass from $u$ to $v$}\},
\end{equation}
where $T(\pi)$ denotes the time needed for the packet to cross the path $\pi$. 
Therefore, the relative network delay from node $u$ to $v$ can be defined as follows~\cite{Liggett}.
\begin{definition}
\label{DelayDef}
The relative network delay of delivering a packet from node $u$ to $v$ in the network $G(\lambda, r_0, g)$ is
\begin{equation}
\gamma(\lambda) = \lim_{d(u,v) \to \infty } \frac{T_\lambda(u, v)}{d(u, v)},
\end{equation}
where $u$ and $v$ belong to the giant component $\mathcal{C}(G(\lambda, r_0, g))$ and $d(u, v)$ measures the Euclidean distance between nodes $u$ and $v$. 
\end{definition} 
\begin{figure*}[!h]
        \centering
        \subfigure[]{%
		\includegraphics[width=1.3in]{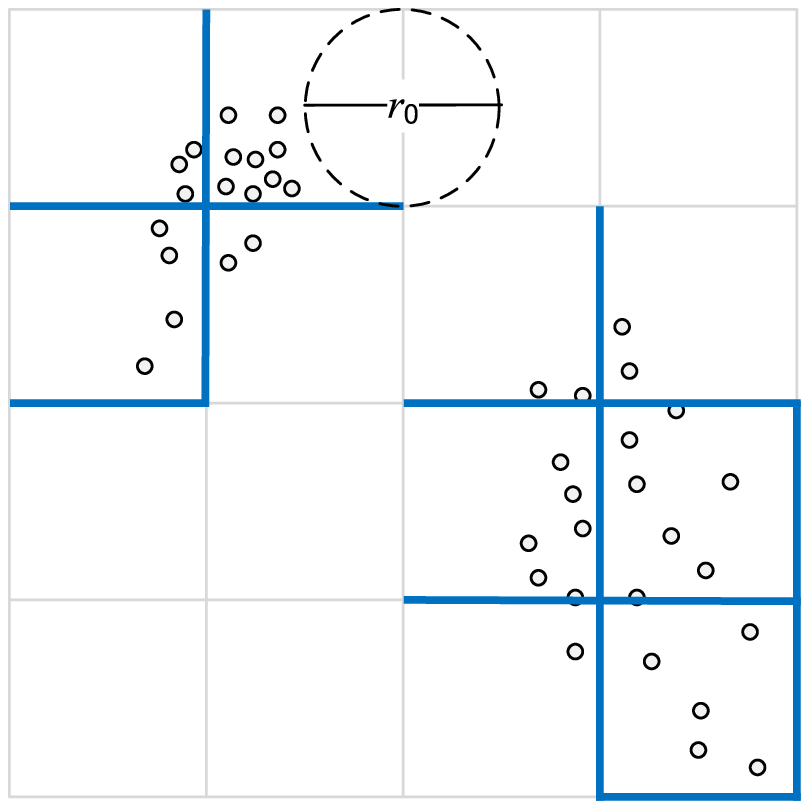}%
		\label{fig:lattice:a}%
		}
		\subfigure[]{%
		\includegraphics[width=1.3in]{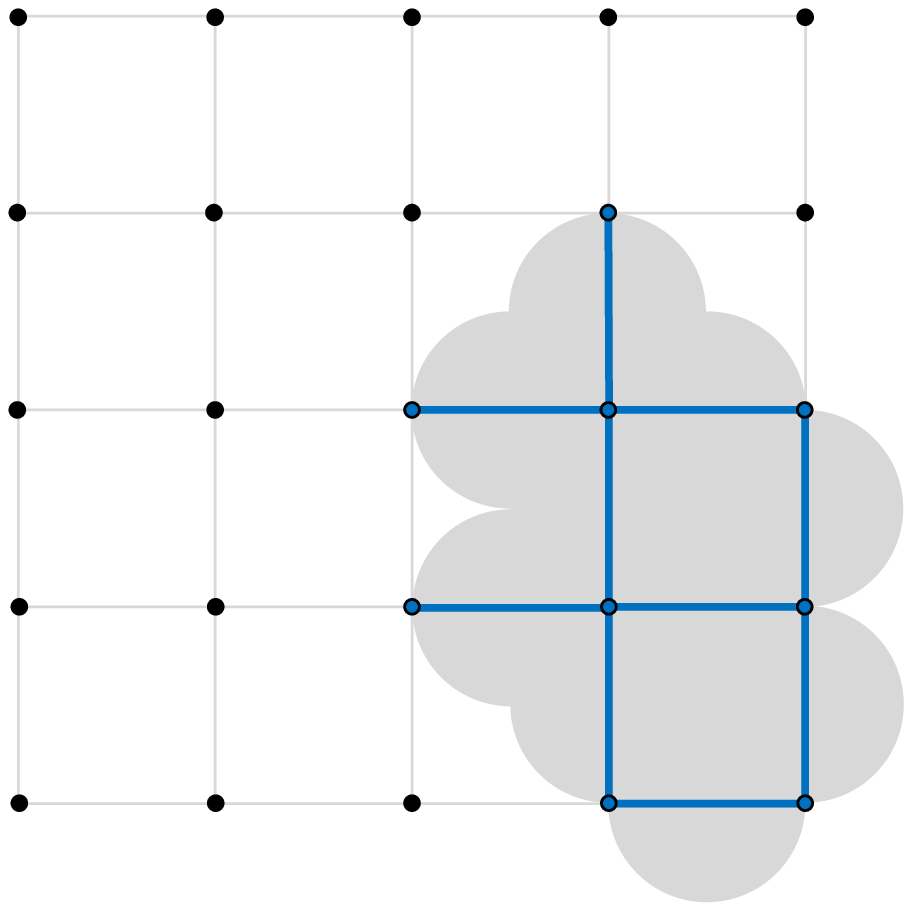}%
		\label{fig:lattice:b}%
		}
		\subfigure[]{%
		\includegraphics[width=1.3in]{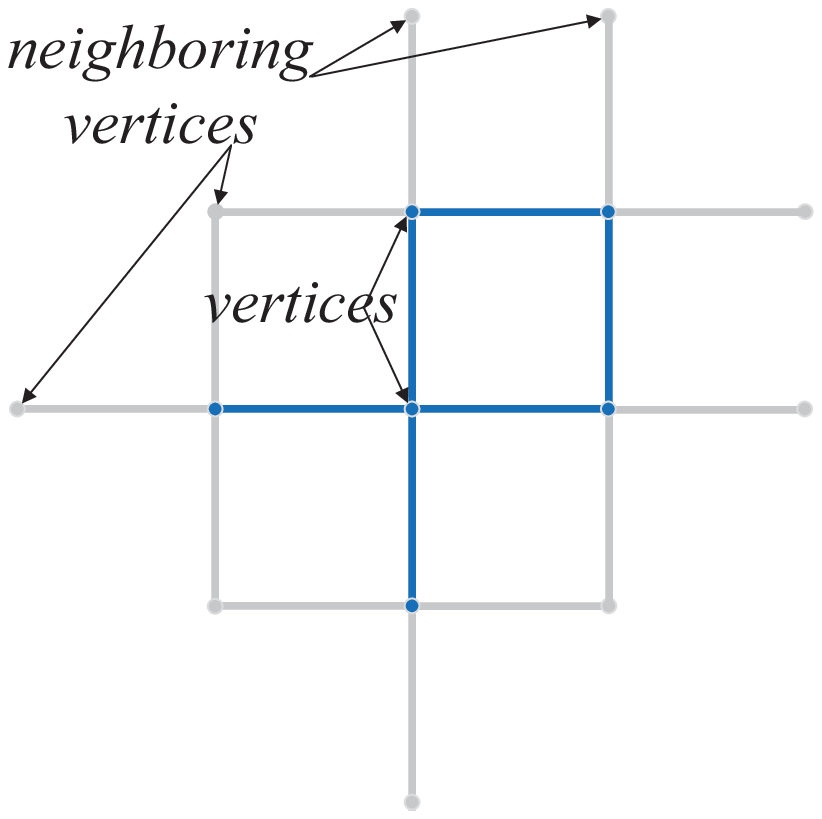}%
		\label{fig:component:a}%
		}
		\subfigure[]{%
		\includegraphics[width=1.3in]{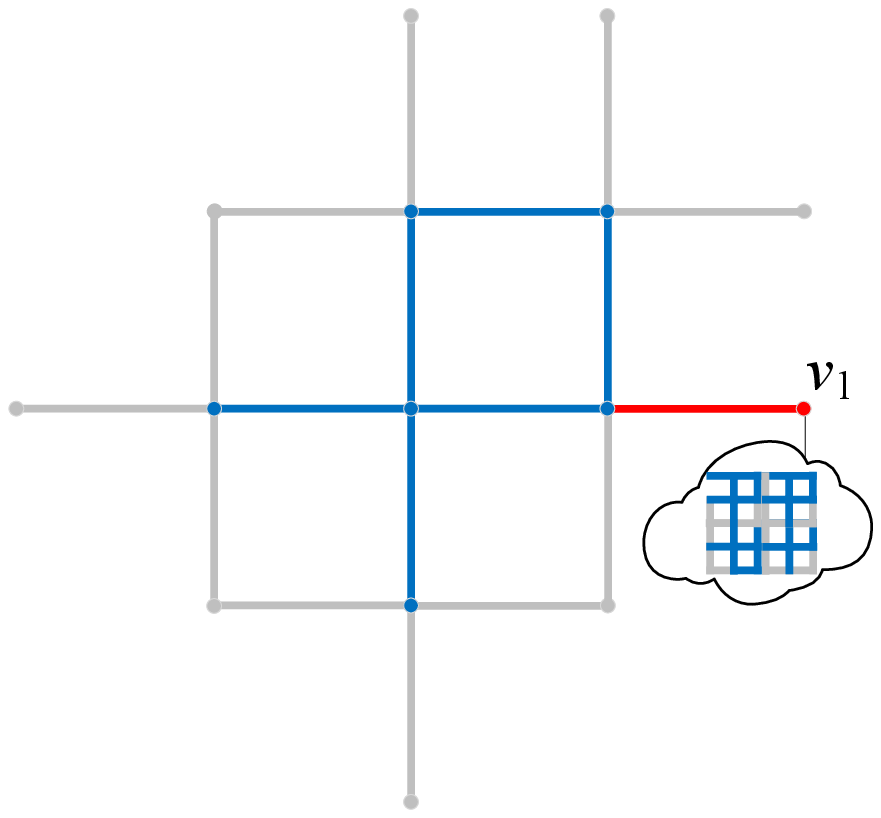}%
		\label{fig:component:b}%
		}	
		\subfigure[]{%
		\includegraphics[width=1.3in]{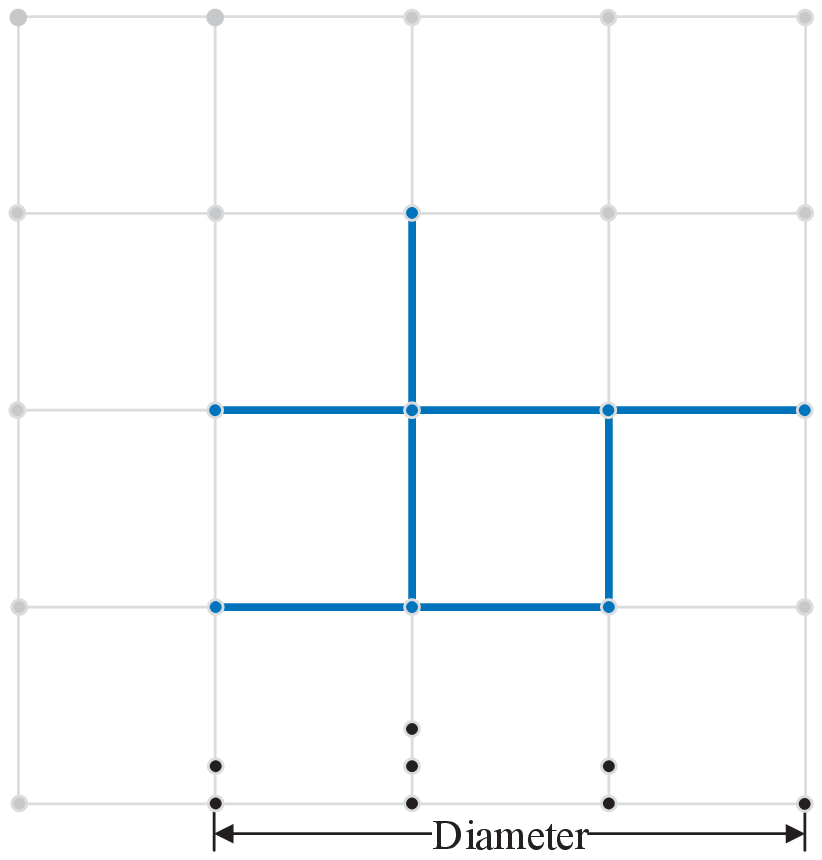}%
		\label{fig:vhedges}%
		}
        \caption{Coupling a network with a lattice. (a) Clusters in the network and corresponding connected components in the lattice. (b) The area of a connected component in the lattice. (c) An $n$-size connected component. (d) A connected component with size of $\geq n+1$. (e) Vertices from a connected component are procjected on the horizontal direction.}\label{fig:lattice}
\end{figure*}

\section{Connected Component Size}
\label{sec:ccsize}
To compute the expected cluster size in an MEHWN, we consider a square lattice $\mathcal{L}$ on top of the network $G(\lambda, r_0, g)$, which is shown in Fig.~\ref{fig:lattice:a}. The length of each edge in the lattice is $r_0$. 

\begin{definition}
\label{Occupy}
Given an edge $e_j$ in the lattice, a circle $S_{e_j}$ can be drawn with the $e_j$ as $S_{e_j}$'s diameter. The edge $e_j$ is called \textbf{occupied} if there is at least one node in the circle; otherwise, it is \textbf{not occupied}.
\end{definition}

As node deployment follows the Poisson Point Process, the probability that an edge is occupied is
\begin{equation}
\label{eq:p}
p= 1- e^{-\lambda \sqrt g \frac{\pi r_0^2}{4}},
\end{equation}
where $\frac{\pi r_0^2}{4}$ is the area of the circle, and $g$ is the probability that a link exists between two nodes. 
Because $\sqrt g < 1$ denotes the probability that a node collects enough energy to be active, the real node density in an MEHWN is smaller than $\lambda$. 
According to the Thinning Theorem~\cite{thinnigThm}, the node density in an MEHWN is reduced from the original $\lambda$ to $\lambda \sqrt g$.

\begin{definition}
\label{ConnectedComponent}
In the lattice $\mathcal{L}$, a set of connected and occupied edges are defined as a \textbf{connected component}, denoted as $\mathcal{C}(\mathcal{L})$.
\end{definition}

Given the probability $p$ of an edge being occupied, we could find several connected components in the lattice. On the other hand, given a cluster in the network, if all the corresponding occupied edges are collected, it ends up with a connected component. Note that for any cluster in network, there must be a unique corresponding connected component in the lattice. However, a connected component in the lattice may correspond to a node, a cluster, or a few adjacent clusters in the network. In summary, connected components provide a larger estimation for clusters in the network.

\begin{definition}
\label{Area}
The \textbf{area} of a connected component $\mathcal{C}(\mathcal{L})$ is defined as
\[
\displaystyle \mathcal{A} (\mathcal{C}(\mathcal{L})) = \cup_{e_j} S_{e_j},
\]
where $\{ e_j \}$ are the edges composing the connected component.
\end{definition}

The shaded area in Fig.~\ref{fig:lattice:b} gives the area of an example connected component, depicted by blue lines. Based on this definition, we know $\mathcal{A}(\mathcal{C}(\mathcal{L}))$ perfectly covers the corresponding cluster(s) in $G(\lambda, r_0, g)$ from the following lemma. 

\begin{lemma}
\label{lm:lattice}
Given the RCM $G(\lambda, r_0, g)$ and a corresponding square lattice $\mathcal{L}$, $\mathcal{A} (\mathcal{C}(\mathcal{L}))$ provides the maximum area where nodes in the underlying cluster(s) could reside. 
\end{lemma}

\begin{proof}

Assume to the contrary that there is a node in the underlying cluster(s) that is located outside $\mathcal{A}(\mathcal{C}(\mathcal{L}))$.
Suppose the node is located in $S_{e_j}$ that corresponds to an occupied edge $e_j$. 
Based on our assumption, $e_j$ is not connected to $\mathcal{C}(\mathcal{L})$.
That also means none of $e_j$'s six neighboring edges belongs to $\mathcal{C}(\mathcal{L})$; otherwise, $e_j$ can be connected to $\mathcal{C}(\mathcal{L})$ through that edge.
That is to say, the node is more than $r_0$ away from the any node in the underlying cluster, contradicting our assumption that the node is within the cluster.
Therefore, the lemma is proved.

\end{proof}

As a connected component provides a larger estimate of the underlying cluster(s), it is important to know the expected size of a connected component. We define the size of a connected component as follows. 
\begin{definition}
The \textbf{size} of a connected component $\mathcal{C}(\mathcal{L})$ is defined as the number of vertices in it, denoted by $\mathcal{S}(\mathcal{C}(\mathcal{L}))$.
\end{definition}

According to the definition, there exist $n$ vertices in an $n$-size connected component that are connected to each other.
We call these vertices as \textit{connected vertices}.
For each vertex in the connected component, there must be zero to three non-connected vertices around it.
We call these non-connected vertices \textit{neighboring vertices}. 
For example, Fig.~\ref{fig:component:a} shows a connected component with $6$ connected vertices and $9$ neighboring vertices. 
It is clear that an $n$-size connected component contains at most $(2n+2)$ neighboring vertices, corresponding to a linear topology.

We use $\Pr\{\mathcal{S}(\mathcal{C}(\mathcal{L})) = n\}$ to denote the probability an arbitrary connected component has the size of $n$.
For the sake of simplicity, we use $\mathcal{S}$ to denote $\mathcal{S}(\mathcal{C}(\mathcal{L}))$.
Then, we can use $\Pr\{\mathcal{S} \geq n\}$ to denote the probability that an arbitrary connection component has the size of at least $n$. 

Given an $n$-size connected component, if it connects to one of its neighboring vertices, a new connected component with a size greater or equal to $n+1$ will be generated. 
This is because the newly connected vertex, e.g., $v_1$ in Fig.~\ref{fig:component:b}, may belong to another connected component.
This important observation will be used to compute the upper bound of $\Pr\{\mathcal{S} =n\}$.

\begin{lemma}
The probability of $\Pr\{\mathcal{S} = n\}$ satisfies
\begin{equation}
 \Pr\{\mathcal{S} = n\} \leq p \prod_{k=3}^n \frac{ 2kp }{2kp + p +1} .
\end{equation}
\end{lemma}

\begin{proof}

For an $(n-1)$-size connected component $\mathcal{C}_{n-1}^i$, we use $s_{n-1}^i$ to denote the number of neighboring vertices it has. We assume it occurs in lattice $\mathcal{L}$ with the probability of $\Pr\{\mathcal{C}_{n-1}^i\}$.
Therefore, we have $\Pr \{\mathcal{S}=n-1 \} = \sum_i \Pr\{\mathcal{C}_{n-1}^i\}$.

In $\mathcal{C}_{n-1}^i$, we randomly select a neighboring vertex and connect it to $\mathcal{C}_{n-1}^i$, resulting in a connected component with size $\geq n$. That implies 
\[
\Pr\{\mathcal{S} \geq n\} =  \sum_i \Pr\{\mathcal{C}_{n-1}^i\} p s_{n-1}^i ,
\]
where $\Pr\{\mathcal{S} \geq n\}$ is the probability that an arbitrary connected component has the size of at least $n$. Similarly, we have 
\[
\Pr\{\mathcal{S} \geq n+1\} =  \sum_j \Pr\{\mathcal{C}_{n}^j\} p s_{n}^j ,
\]
where $s_n^j$ is the number of neighboring vertices of $\mathcal{C}_{n}^j$.
Subtracting these two equations, we have
\begin{equation}
\label{eq:psn}
\begin{split}
\Pr\{\mathcal{S} = n\} &= \Pr\{\mathcal{S} \geq n\} - \Pr\{\mathcal{S} \geq n+1\} \\
	&=  \sum_i \Pr\{\mathcal{C}_{n-1}^i\} p s_{n-1}^i -  \sum_j \Pr\{\mathcal{C}_{n}^j\} p s_{n}^j.
\end{split}
\end{equation}

For an $n$-size connected component, e.g., $\mathcal{C}_{n}^j$, we group it based on the corresponding $(n-1)$-size connected component from which it can be generated.
Without loss of generality, we assume $\mathcal{C}_{n}^j$ is generated from  $\mathcal{C}_{n-1}^i$, i.e., removing the vertex with the largest number of neighboring vertices from $\mathcal{C}_{n}^j$ will yield $\mathcal{C}_{n-1}^i$.
As such, we re-label $\mathcal{C}_{n}^j$ as $\mathcal{C}_{n}^{ij}$ indicating it belongs to the $i$th group.
We further use $\{\mathcal{C}_{n}^{ij}\}$ to denote all the  $n$-size connected components generated from $\mathcal{C}_{n-1}^i$.
It is possible that a $\mathcal{C}_{n}^j$ is generated from several $\mathcal{C}_{n-1}^i$'s. After the grouping, however, there must be exactly $| \{ \mathcal{C}_{n-1}^i \}|$ number of groups.
Therefore, Eq.~\ref{eq:psn} can be rewritten as
\[
\Pr\{\mathcal{S} = n\} = \sum_i \Pr \{ \mathcal{C}_{n-1}^i \} p s_{n-1}^i - \sum_i \Pr \{ \mathcal{C}_{n}^{ij} \} p s_{n}^{ij},
\]
where $\Pr \{ \mathcal{C}_{n}^{ij} \}$ is the sum of probabilities of all $n$-size connected components in the $i$th group.
Because $\mathcal{C}_{n}^j$ is generated from $\mathcal{C}_{n-1}^i$, it must contain at least $(s_{n-1}^i + 1) $ number of neighboring vertices. 
Therefore, we have 
\[
\Pr\{\mathcal{S} = n\} \leq
\sum_i \Pr \{ \mathcal{C}_{n-1}^i \} p s_{n-1}^i - \sum_i \Pr \{ \mathcal{C}_{n}^{ij} \} p (s_{n-1}^i +1),
\]
which can be rewritten as
\begin{equation*}
\begin{split}
 \Pr\{\mathcal{S} = n\} \leq  & \left( \Pr \{ \mathcal{S} = n-1\} - \Pr \{ \mathcal{S} = n\}\right) p s_{n-1}^i\\
& - p \Pr \{ \mathcal{S} = n\} \\
\end{split}
\end{equation*}

Because the maximum value of $s_{n-1}^i$ is $2n$, we have 
\begin{equation*}
\begin{split}
 \Pr\{\mathcal{S} = n\} \leq  & 2np \left( \Pr \{ \mathcal{S} = n-1\} - \Pr \{ \mathcal{S} = n\}\right)\\
& - p \Pr \{ \mathcal{S} = n\} \\
\end{split}
\end{equation*}
which yields
\[
\Pr\{\mathcal{S} = n\} \leq \frac{ 2np }{2np + p +1} \Pr\{\mathcal{S} = n-1\}.
\]

We know that probability $\Pr\{\mathcal{S} = 2\} \leq p$ where $p$ is the probability that two adjacent vertices are connected, so we have 

\begin{equation}
\Pr\{\mathcal{S} = n\} \leq p \prod_{k=3}^n \frac{ 2kp }{2kp + p +1} .
\end{equation}

\end{proof}

We denote the upper bound of $\Pr\{\mathcal{S} = n\}$ as $\bar p_n$, then the expected connected component size satisfies
\begin{equation}
\label{eq:componentsize}
E(\mathcal{S}(\mathcal{C}(\mathcal{L}))) \leq \sum_{n=1}^{\infty} n \bar p_n. 
\end{equation} 

\section{Connected Component Diameter}
\label{sec:ccdiameter}
\begin{definition}
\label{Def:diameter}
The \textbf{diameter} of a connected component $\mathcal{C}(\mathcal{L})$ is defined as its horizontal distance in terms of number of edges, which is denoted by $\mathcal{D}(\mathcal{C}(\mathcal{L}))$.
\end{definition}

For an $n$-size connected component $\mathcal{C}_n(\mathcal{L})$, we know there are $n$ connected vertices. 
If these $n$ vertices are projected onto the horizontal direction in the lattice, there will be 1 to $n$ points generated along the horizontal direction.  
If all vertices are projected onto one point, the connected component's diameter will be zero; if all vertices are projected onto different points, the connected component's diameter will be $n-1$.
Based on this observation, we introduce the following lemma to compute the upper bound of the diameter of $\mathcal{C}_n(\mathcal{L})$.

\begin{lemma}
\label{lm:pk}
Given an $n$-size connected component $\mathcal{C}_n(\mathcal{L})$, the probability that $\mathcal{D}(\mathcal{C}_n(\mathcal{L})) = k$ satisfies
\begin{equation}
\label{eq:pk}
p_k(n) \leq \sum_{a=k}^{n-1} C_{n-1}^{a} \left( \frac{1}{2}\right)^{n-1} \left( \frac{1}{k}\right)^{a-k}.
\end{equation}
\end{lemma}

\begin{proof}
From the connected component $\mathcal{C}_n(\mathcal{L})$, we randomly pick one vertex. 
Suppose the vertex is projected to a point located at $x_0$ on the horizontal direction. 
The other $n-1$ vertices could be projected onto $x_0$ or not, each with a probability $0.5$.
If $a$ number of vertices are projected to points that are different from $x_0$, then the connected component's diameter will range from 1 to $a$.
The probability of this event occurring can be computed from
\[
p_a (n) = C_{n-1}^{a}  \left( \frac{1}{2} \right)^{n-1}, 
\]
when $C_{n-1}^{a}$ refers to the combination of $n-1$ vertices taken $a$ at a time without repetition.
Out of $n-1$ vertices, $n-1-a$ vertices are projected to $x_0$ with the probability of $\left( \frac{1}{2} \right)^{n-1-a}$, and $a$ vertices are projected to other points with the probability of $\left( \frac{1}{2} \right)^{a}$. 

Because the diameter of $\mathcal{C}_n(\mathcal{L})$ is $k$, we have $k \leq a \leq n$.
If $k \leq a$, there must be $k$ vertices projected onto $k$ different points (not $x_0$) along the horizontal direction. 
We denote these points as $\{x_1, x_2, \cdots, x_k\}$.
For the other $k-a$ vertices, they must be projected onto the points from $\{x_1, x_2, \cdots, x_k\}$.  
Therefore, the conditional probability of $\mathcal{D}(\mathcal{C}_n(\mathcal{L})) = k$ when $a$ vertices are projected to points other than $x_0$ must satisfy 
\begin{equation*}
\begin{split}
\Pr\{\mathcal{D}(\mathcal{C}_n(\mathcal{L})) = k | & a \mbox{ vertices are projected to points} \\
& \mbox{other than } x_0 \} \leq \left( \frac{1}{k} \right)^{a-k}.
\end{split}
\end{equation*}
Note that $\left( \frac{1}{k} \right)^{a-k}$ provides the upper bound because it is possible that the $a-k$ vertices are not connected to the other $k$ vertices in the lattice. Therefore, we have
\begin{equation}
\label{eq:pk}
\begin{split}
p_k (n) &= \sum_{a=k}^{n-1} p_a(n)  \Pr\{\mathcal{D}(\mathcal{C}_n(\mathcal{L})) = k | a \mbox{ vertices are} \\
& \mbox{ projected to points other than } x_0  \}  \\
    &\leq \sum_{a=k}^{n-1} C_{n-1}^{a}  \left( \frac{1}{2} \right)^{n-1} \left( \frac{1}{k} \right)^{a-k}.
\end{split}
\end{equation}
\end{proof}

Based on Lemma~\ref{lm:pk}, the expected diameter of an $n$-size connected component satisfies
\begin{equation}
\label{eq:diameterupper}
\begin{split}
E(\mathcal{D}(\mathcal{C}_n(\mathcal{L}))) & = \sum_{k=1}^{n} k p_k(n) \\
 &\leq \sum_{k=1}^{n-1} k\sum_{a=k}^{n-1} C_{n-1}^{a}  \left( \frac{1}{2} \right)^{n-1} \left( \frac{1}{k} \right)^{a-k}
\end{split}
\end{equation}
\section{Lower and Upper Bounds of Relative Network Delay}
\label{sec:bounds}

\subsection{Lower Bound}
With the upper bounds of $p_n$ and $E(\mathcal{D}(\mathcal{C}_n(\mathcal{L})))$, we are able to compute the lower bound of the relative network delay $\gamma(\lambda)$. Before proceeding, we will introduce the following lemma.

\begin{lemma}
\label{lm:diameter}
Given the RCM $G(\lambda, r_0, g)$ with $\lambda_L < \lambda < \lambda_I$, $\gamma(\lambda)$ satisfies
\begin{equation}
\gamma(\lambda) \geq \frac{1}{E(\mathcal{D}_g(\lambda) + 1) r_0}
\end{equation}
where $\mathcal{D}_g(\lambda)$ is the diameter of a connected component in the lattice.
\end{lemma}

This lemma is from~\cite{Wang11} with subtle difference where the lower bound is changed from $\frac{1}{E(S_g(\lambda) + r_0)}$ to $ \frac{1}{E(\mathcal{D}_g(\lambda) + r_0)}$ where $S_g(\lambda)$ represents the cluster size in~\cite{Wang11} while $\mathcal{D}_g(\lambda)$ is the connected component diameter in this paper.

\begin{proof}
Given two nodes $u$ and $v$, we can find a series of clusters connecting $u$ to $v$ 
\[ \{ (u_1, t_1), (u_2, t_2), \cdots, (u_M, t_M) \},\]
where $(u_k, t_k)$ indicates the cluster originating from node $u_k$ at time slot $t_k$. 
Therefore, we have $u_1 = u$, $u_M$ and $v$ are in the same cluster. 
The delay of transmitting a packet via these clusters is
\[ \sum_{k=1}^{M-1}(t_{k+1} - t_{k}) \geq M-1, \]
because clusters $(u_k, t_k)$ and $(u_{k+1}, t_{k+1})$ appears in two different time slots, \ie, $t_{k+1} - t_k \geq 1$.

At a certain time slot $t_k$, we assume $u_k$ is connected (by a multi-hop path) to a node $u_k'$ which is later connected to $u_{k+1}$ at time $t_{k+1}$. Taking node $u$ as the origin, we draw a line from $u$ to $v$ and consider this line the horizontal axis of the new coordinate system. Then, we have
\[
\begin{split} 
d(u_k, u_{k+1}) & \leq d(u_k, u'_k) + d(u'_k, u_{k+1})  \\
& \leq S_{g,t_k, u_k}(\lambda) + r_0,
\end{split}
\]
where $S_{g,t_k, u_k}(\lambda)$ denotes the cluster size of $(u_k, t_k)$. Considering only the horizontal distances among $u_k$, $u'_k$, and $u_{k+1}$, we have
\[
\begin{split} 
d_x(u_k, u_{k+1}) & \leq d_x(u_k, {u'}_k) + d_x({u'}_k, u_{k+1}) \\
& \leq S_{x,g,t_k, u_k}(\lambda) + r_0,
\end{split}
\]
where $d_x(\cdot)$ measure the horizontal distance between two nodes. $S_{x,g,t_k, u_k}(\lambda)$ is the diameter of cluster $(u_k, t_k)$, \ie, the largest horizontal distance between nodes in $(u_k, t_k)$. Therefore, we have
\[
\begin{split}
d_x(u, v) & \leq \sum_{k=1}^{M-1} (S_{x, g,t_k, u_k}(\lambda) + r_0) + S_{x, g,t_M, u_M}(\lambda) \\
& < \sum_{k=1}^{M}(S_{x, g,t_k, u_k}(\lambda) + r_0).
\end{split}
\]

For each $k$, $S_{x,g,t_k, u_k}(\lambda)$ admits the same distribution, so we rewrite it as $S_{x,g}(\lambda)$.
According to the Lemma 5 in~\cite{wangton}, we have
\[
\lim_{M\rightarrow \infty} \frac{\sum_{k=1}^{M}(S_{x, g,t_k, u_k}(\lambda) + r_0)}{M} = E(S_{x, g}(\lambda) + r_0).
\]
Note that this result cannot be derived from the strong law of large number because $S_{x,g,t_k, u_k}(\lambda)$ may not be independent for different $k$'s~\cite{wangton}.

For a small number $\epsilon > 0$, $\exists M_1$ such that $\forall M > M_1$, we have
\[
\frac{\sum_{k=1}^{M}(S_{x, g,t_k, u_k}(\lambda) + r_0)}{M} = E(S_{x,g}(\lambda) + r_0) \pm \epsilon.
\]
Therefore,
\[
d_x(u, v) < \sum_{k=1}^{M} (S_{x,g}(\lambda) + r_0 ) < M(E(S_{x,g}(\lambda)+r_0) \pm \epsilon).
\]

Because $\mathcal{D}_g(\lambda)$ is the diameter of a connected component in the lattice, we have $\mathcal{D}_g(\lambda) r_0 \geq S_{x, g}(\lambda)$, due to Lemma~\ref{lm:lattice}. Then,
\[
\begin{split}
M-1 & > \frac{d_x(u, v)}{E(S_{x,g}(\lambda) + r_0) \pm \epsilon} -1 \\
&\geq \frac{d_x(u, v)}{E(\mathcal{D}_g(\lambda) + 1)r_0 \pm \epsilon} -1.
\end{split}
\]

We are interested in the relative network delay, i.e., when $d(u, v) \rightarrow \infty$; therefore, we assume $u$ is far away from $v$ and we have
\[
T_{\lambda}(u, v) \geq M-1 > \frac{d_x(u, v)}{E(\mathcal{D}_g(\lambda) + 1)r_0 \pm \epsilon} -1.
\]
Therefore, when $\epsilon \rightarrow 0$
\[
\begin{split}
\gamma(\lambda) & = \lim_{d(u, v) \rightarrow \infty} \frac{T_{\lambda}(u, v)}{d(u, v)}\\
& = \lim_{d_x(u, v) \rightarrow \infty} \frac{T_{\lambda}(u, v)}{d_x(u, v)} \\
& \geq \frac{1}{E(\mathcal{D}_g(\lambda) + 1)r_0}.
\end{split}
\]
\end{proof}

Due to the above lemma, we are able to identify the lower bound of $\gamma(\lambda)$ in the following theorem.
\begin{theorem}
Given the RCM $G(\lambda, r_0, g)$ with $\lambda_L < \lambda < \lambda_I$, $\gamma(\lambda)$ satisfies
\begin{equation}
\begin{split}
& \gamma(\lambda) \geq \\
& \frac{1}{\left(  \sum_{n=2}^{\infty} \bar p_n \sum_{k=1}^{n-1} k\sum_{a=k}^{n-1} C_{n-1}^{a}  \left( \frac{1}{2} \right)^{n-1} \left( \frac{1}{k} \right)^{a-k} +1 \right) r_0}, \\
\end{split}
\end{equation}
where $\bar p_n$ is the upper bound of the probability that a connected component's size is $n$.
\end{theorem}

\begin{proof}
The expected connected component's diameter in $\mathcal{L}$ satisfies
\begin{equation*}
E(\mathcal{D}_g(\lambda)) \leq  \sum_{n=1}^{\infty} E(\mathcal{D}(\mathcal{C}_n(\mathcal{L})))  \bar p_n,
\end{equation*}
Based on the inequation~\ref{eq:diameterupper}, we get
\begin{equation*}
\begin{split}
E(\mathcal{D}_g(\lambda)) \leq & \sum_{n=2}^{\infty} \bar p_n \sum_{k=1}^{n-1} k\sum_{a=k}^{n-1} C_{n-1}^{a}  \left( \frac{1}{2} \right)^{n-1} \left( \frac{1}{k} \right)^{a-k}\\
&+E(\mathcal{D}(\mathcal{C}_1(\mathcal{L}))) \bar p_1.
\end{split}
\end{equation*}
Because the diameter of a connected component with one vertex is always zero, we have
\begin{equation}
\label{eq:expdiameter}
E(\mathcal{D}_g(\lambda)) \leq \sum_{n=2}^{\infty} \bar p_n \sum_{k=1}^{n-1} k\sum_{a=k}^{n-1} C_{n-1}^{a}  \left( \frac{1}{2} \right)^{n-1} \left( \frac{1}{k} \right)^{a-k}.
\end{equation}
Due to Lemma~\ref{lm:diameter}, we can get the lower bound of $\gamma(\lambda)$ from the following inequation.
\begin{equation*}
\begin{split}
& \gamma(\lambda) \geq \\
& \frac{1}{\left(  \sum_{n=2}^{\infty} \bar p_n \sum_{k=1}^{n-1} k\sum_{a=k}^{n-1} C_{n-1}^{a}  \left( \frac{1}{2} \right)^{n-1} \left( \frac{1}{k} \right)^{a-k} + 1 \right) r_0}, \\ 
\end{split}
\end{equation*}
where $\bar p_n$ is the upper bound of $\Pr\{\mathcal{S}=n\}$.
\end{proof}

Note that we find a closed-form expression of the lower bound of $\gamma(\lambda)$.
It is different from the approximated one, \ie, $\gamma(\lambda) \geq \frac{1}{E(S_g(\lambda)) + r_0}$ in~\cite{Wang11} where $E(S_g(\lambda))$ is approximated as $\frac{1.2841 \lambda}{2.4886 - \lambda}$.

\subsection{Upper Bound}
To compute the upper bound of $\gamma(\lambda)$ in an MEHWN, we first construct a sparse network from $G(\lambda, r_0, g)$ by randomly removing $(1-\lambda_L / \lambda)$ portion of nodes. 
According to the Thinning Theorem\cite{thinnigThm}, the resulting graph consisting of all remaining nodes along with their associated links can still be modeled by the RCM. 
Then, we compute the $\gamma(\lambda_L)$ of any two nodes in the sparse network. 
Finally, we calculate the upper bound of $\gamma(\lambda)$ in the original network.

\begin{theorem}
\label{thm1}
Given the RCM $G(\lambda, r_0, g)$ with $\lambda_L < \lambda < \lambda_I$, $\gamma(\lambda)$ satisfies
\begin{equation}
\gamma(\lambda) \leq \gamma (\lambda_L).
\end{equation}
\end{theorem}

The proof of this theorem could be found in the Appendix. The above upper bound of $\gamma(\lambda)$ can be expressed as
\begin{equation}
\label{eq:myupper}
\gamma (\lambda_L) = \lim_{d(w_1,w_2) \to \infty } \frac{T_{\lambda_L}(w_1, w_2) E[T(e)] }{d(w_1, w_2)} = \kappa  E[T(e)],
\end{equation}
where $\kappa = \displaystyle \lim_{d(w_1,w_2) \to \infty } \frac{T_{\lambda_L}(w_1, w_2) }{d(w_1, w_2)}$.
The existence of $\kappa$ when node density is $\lambda_L$ is proved in~\cite{Wang11}. Because the upper bound of $\gamma(\lambda)$ identified in~\cite{Wang11} is
\[
\gamma(\lambda)\leq \kappa \sqrt{\frac{\lambda}{\lambda_L}} E[T(e)],
\]
and $\kappa  E[T(e)] \leq \kappa \sqrt{\frac{\lambda}{\lambda_L}} E[T(e)] $, we find a tighter upper bound of $\gamma(\lambda)$.


\section{Simulation Results}
\label{sec:simulation}

To conduct a fair comparison with the start-of-art results presented in~\cite{Wang11}, we choose the same simulation parameters used in~\cite{Wang11}.
Specifically, we simulate networks within a $20 \times 20$ square area using MATLAB. 
Nodes are deployed in the area according to the Poisson Point Process, with different $\lambda$'s ranging from 1.4 to 2.8.
Communication range $r_0$ is set to be $1$. 
Different parameters, such as larger network areas and higher node densities, can be used to conduct simulations. 
Because we observed similar results, they are not presented here.

\begin{figure*}[!h]
        \centering  
        \subfigure[]{%
		\includegraphics[width=1.6in]{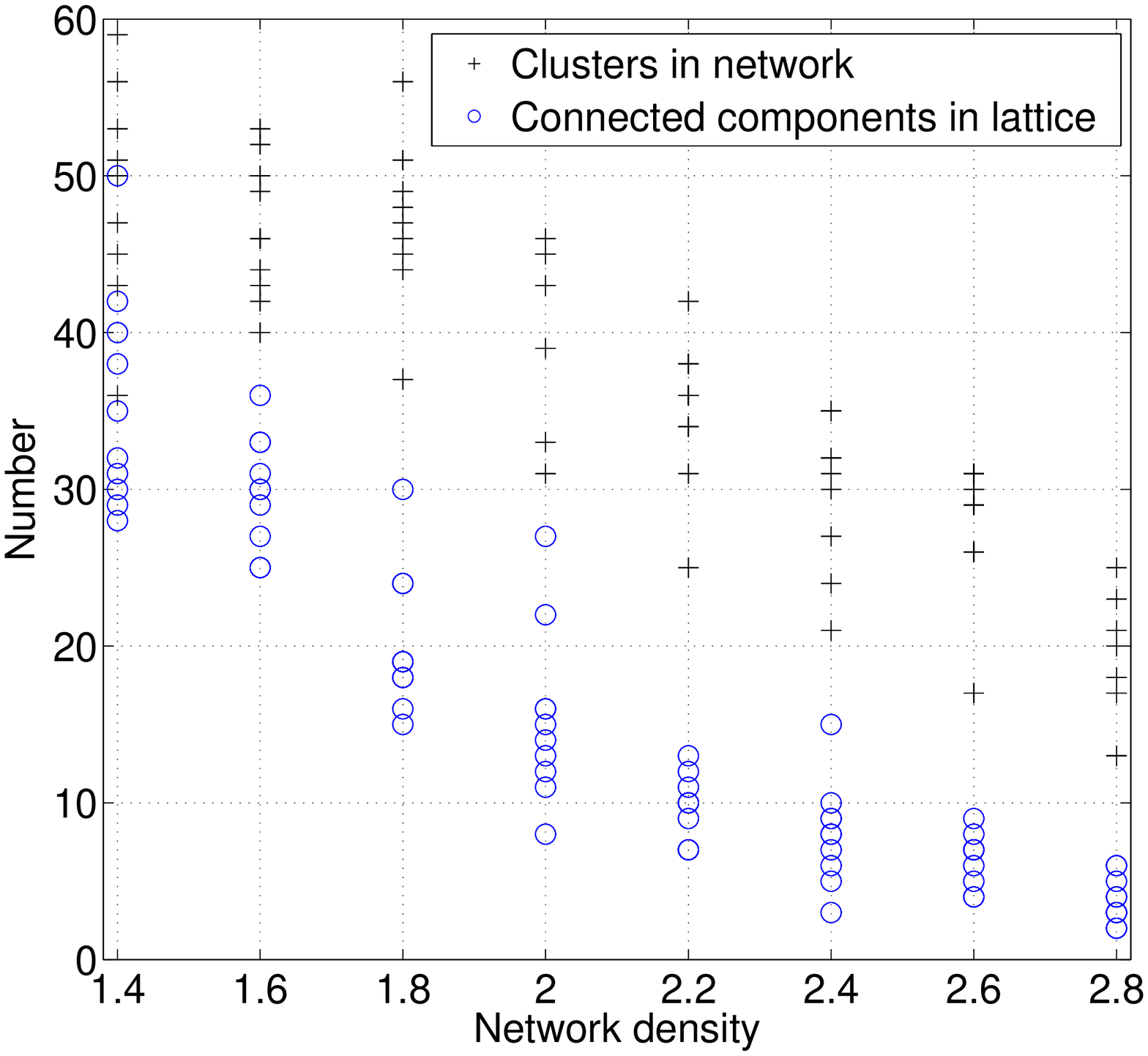}%
		\label{fig:clusternumber}%
		}
 		\subfigure[]{%
		\includegraphics[width=1.6in]{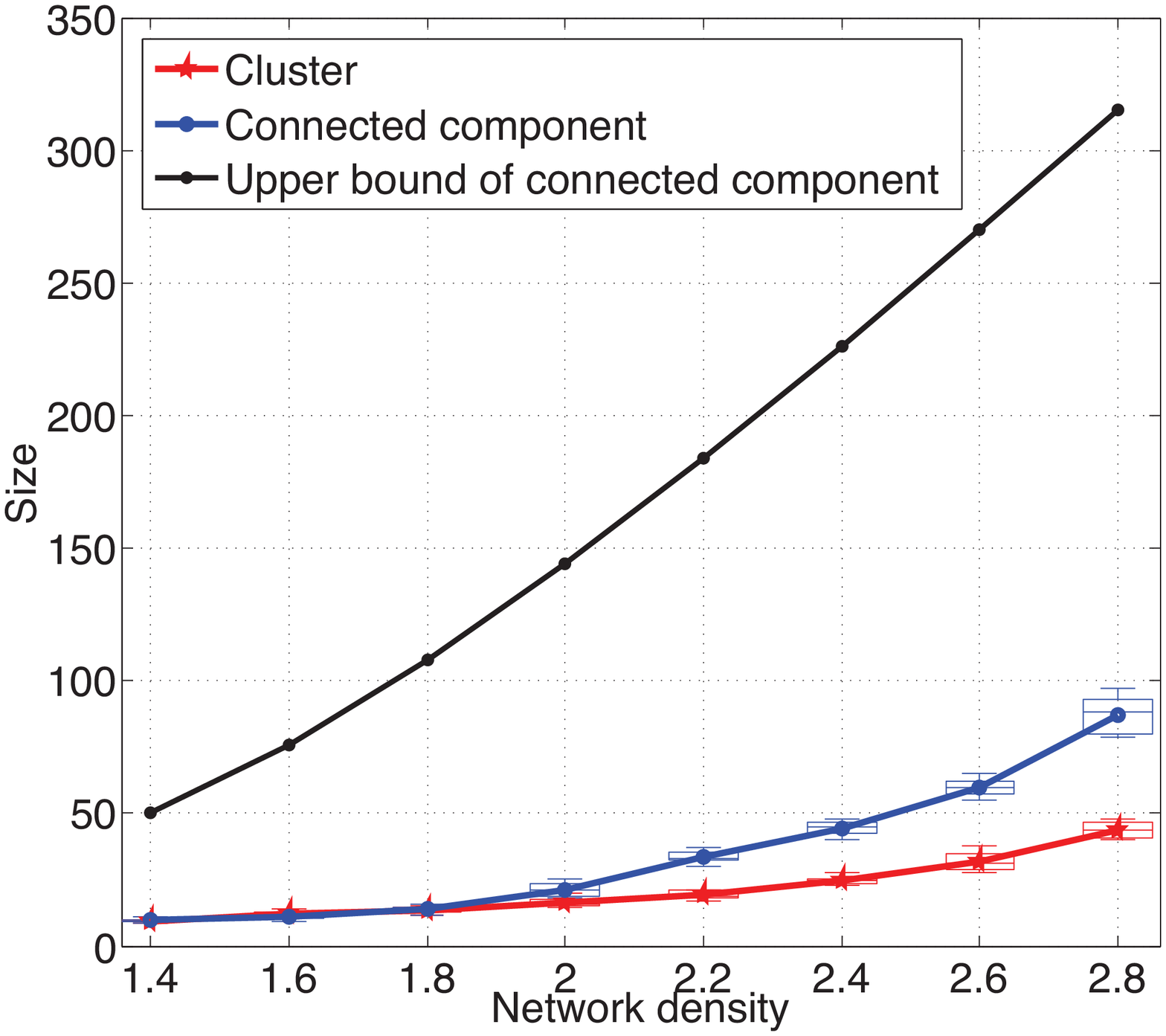}%
		\label{fig:size}%
		}
        \subfigure[]{%
		\includegraphics[width=1.6in]{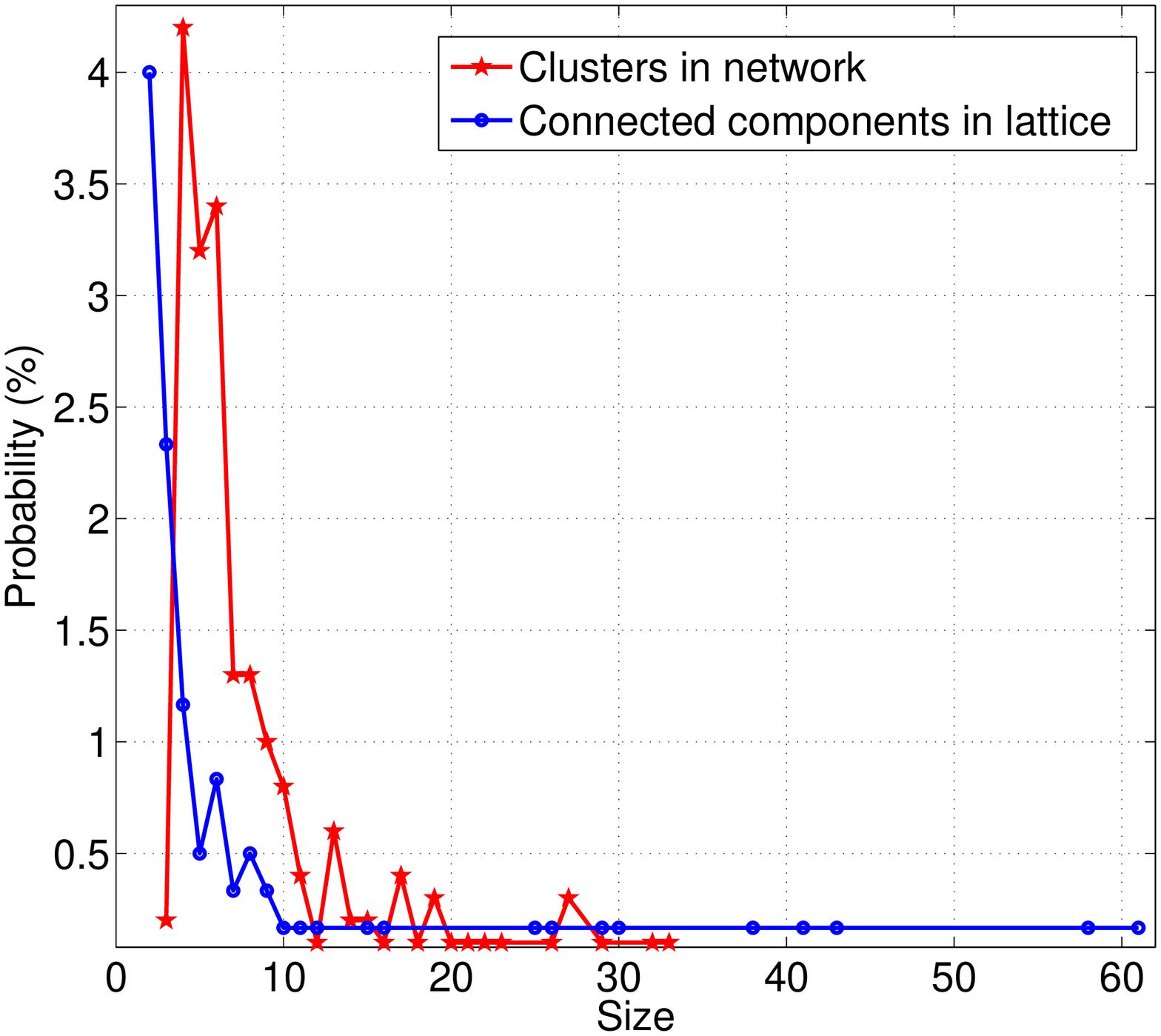}%
		\label{fig:distribution:a}%
		}
		\subfigure[]{%
		\includegraphics[width=1.6in]{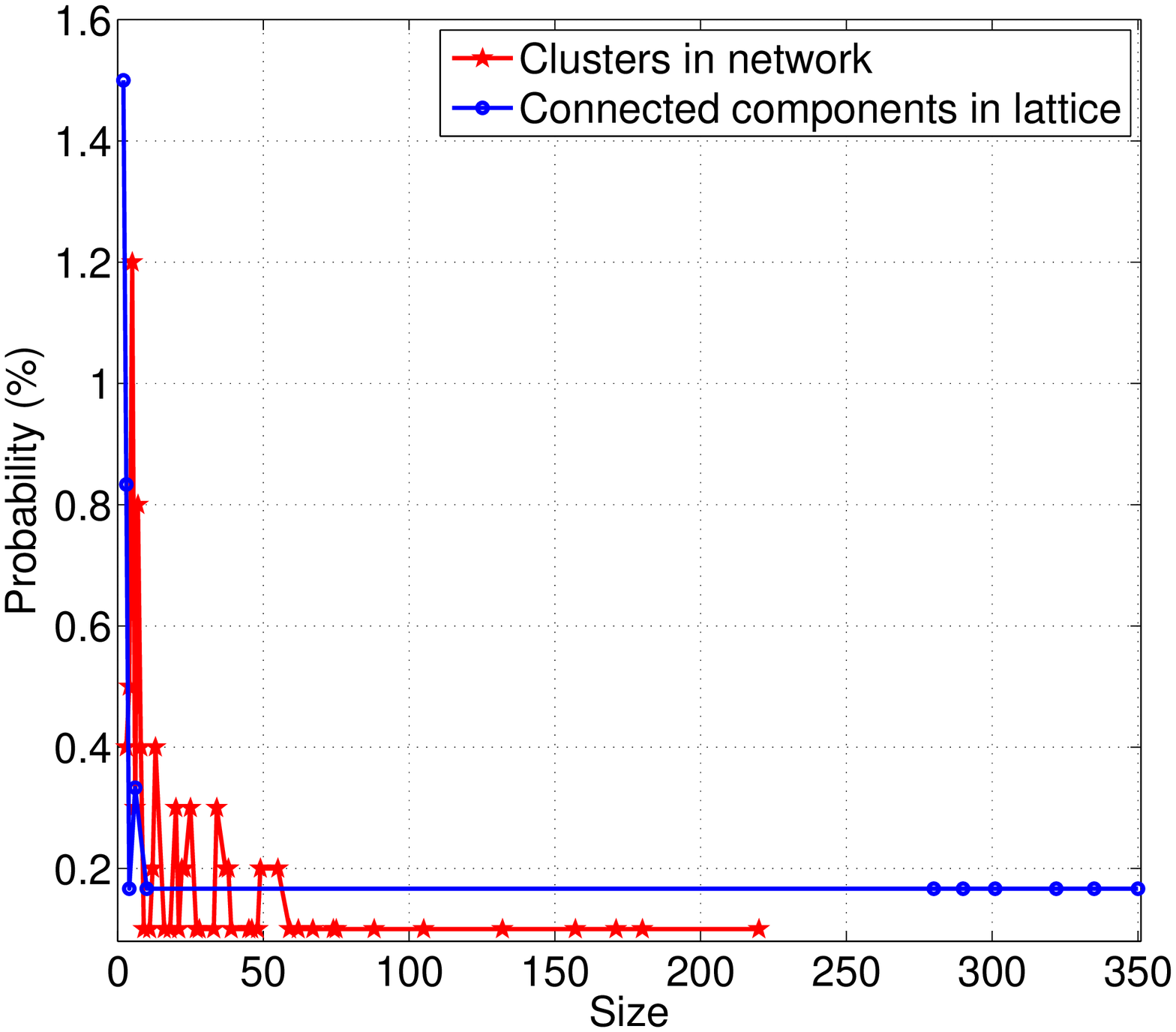}%
		\label{fig:distribution:b}%
		}
        \caption{(a) Number of clusters and number of connected components with various $\lambda$. (b) Cluster sizes and connected component sizes with various $\lambda$'s. (c) PDF of connected component size and cluster size with $\lambda=1.4$. (d) PDF of connected component size and cluster size with $\lambda = 2.8$.}\label{fig:distribution}
\end{figure*}

In simulations, time is broken up into time slots.
In each time slot, a node becomes active with the probability of $q = \sqrt{g}$. 
Two nodes are connected in a time slot if they are active and their distance is smaller or equal to $r_0$.
In this way, several clusters will be generated in each slot.

In the first slot, the source node sends a message to all nodes in the cluster that it belongs to.
If the destination is in the same cluster, we think that the relative network delay is $0$.
Otherwise, nodes that just received this message will forward it to other nodes in the next slot.
The total number of slots needed to forward the message from the source to the destination is considered the network delay.
Dividing this delay by the distance between the source and the destination, we obtain the relative network delay.
%

%
%
%

\subsection{Connected Components vs Clusters}
Keeping $g=0.25$, we change $\lambda$ from $1.4$ to $2.8$ and plot the numbers of clusters and connected components in Fig.~\ref{fig:clusternumber}.
When $\lambda > 1.4$, there are more clusters and less connected components, however, the numbers of clusters and connected components decrease when $\lambda$ increases. 
That makes sense because when $\lambda$ increases to a certain value, both network and lattice percolate with a high probability. 
If the network/lattice percolates, there will be a giant cluster/connected component in the network, \ie, the number of clusters/connected components approaches to 1.

With the same setting, we plot the sizes of clusters and connected components in Fig.~\ref{fig:size}. 
Here the cluster size is measured as the number of vertices in the corresponding connected component. 
As expected, the sizes of both clusters and connected components increase as $\lambda$ increases. 
In addition, we see a potential exponential increase of the sizes, \ie, when $\lambda$ is large enough, the entire network percolates, resulting a giant component. 
This figure also contains the upper bound values of connected component's size that are computed from inequation~\ref{eq:componentsize}. 
We see the upper bound values are always larger than connected component and cluster sizes, verifying our theoretical results.

We further plot the distributions of cluster sizes and connected component sizes in Fig.~\ref{fig:distribution}. 
We see the distribution of cluster sizes is very similar to that of connected component sizes, which confirms our assumption of using connected component to estimate cluster. 
In addition, we find there are more larger connected components in the lattice. 
This phenomena is expected due to Lemma~\ref{lm:lattice}, \ie, the size of connected component should be greater than the that of cluster. 
%

Fig.~\ref{fig:xdiameter} shows the diameters of connected components and clusters with different node densities. 
We see that connected component's diameter is an good estimation of cluster's diameter. 
There is a slight difference between them when $\lambda < 1.1$, however, the difference increases to 4 when $\lambda = 2.8$. 
This figure also implies that expected connected component's diameter can be used to approximate the expected cluster's diameter. 
This observation supports our goal in identifying a closed-form expression of the lower bound of $\gamma(\lambda)$ by computing the upper bound of expected connected component's diameter.

The upper bound of the expected connected component diameter is computed from inequation~\ref{eq:expdiameter} and plotted in Fig.~\ref{fig:xdiameter}. We see that the computed upper bound values are always larger than the expected diameters of connected components and clusters in the network. 

\subsection{Upper and Lower Bounds}
One critical information in the upper bound equation of $\gamma(\lambda)$ is $\kappa$.
We know $\kappa$ is defined as $\frac{N_{\lambda_L}(d)}{d}$ when $d \rightarrow \infty$. 
$\lambda_L$ is the long-term critical network density, so it is equal to $1.44$. 
To obtain this value, we randomly select two nodes in the network, and consider the number of hops along the shortest path between them as $N_{\lambda_L}(d)$. 
We find that the value of $\kappa$ varies, depending on the distance between these two nodes.
However, $\kappa \approx 1.7$ given large amount of simulations.

To accurately compute the value of $\gamma(\lambda)$, we select four pairs of nodes: (1) nodes with the minimum and maximum x-coordinates, (2) nodes with the minimum and maximum y-coordinates, (3) nodes located at the left-bottom and right-top corners of the area, and (4) two nodes that are randomly selected. 
The first three pairs of nodes are selected due to the definition of $\gamma(\lambda)$, which measures the delay-distance ratio when the distance goes to infinity. 
The last pair is randomly selected to avoid selection bias.
In the simulations, we randomly selected 100 pairs of these nodes, so we tested 103 pairs of nodes in total. 

For each pair of nodes, we repeat the simulation $10$ times and present only the average results in the paper. 
In Fig.~\ref{fig:delay}, we plot the average value of $\gamma(\lambda)$ with $\lambda$ ranging from $1.4$ to $2.8$. 
We see that $\gamma(\lambda)$ decreases quickly at the beginning (small $\lambda$), and then slowly, and finally reaches $0$. 
This phenomena makes sense because when $\lambda$ is large, the network is percolated with a high probability. 
If a network percolates, relative network delay is $0$.
\begin{figure*}[]
\centering
		\subfigure[]{%
		\includegraphics[height=1.6in]{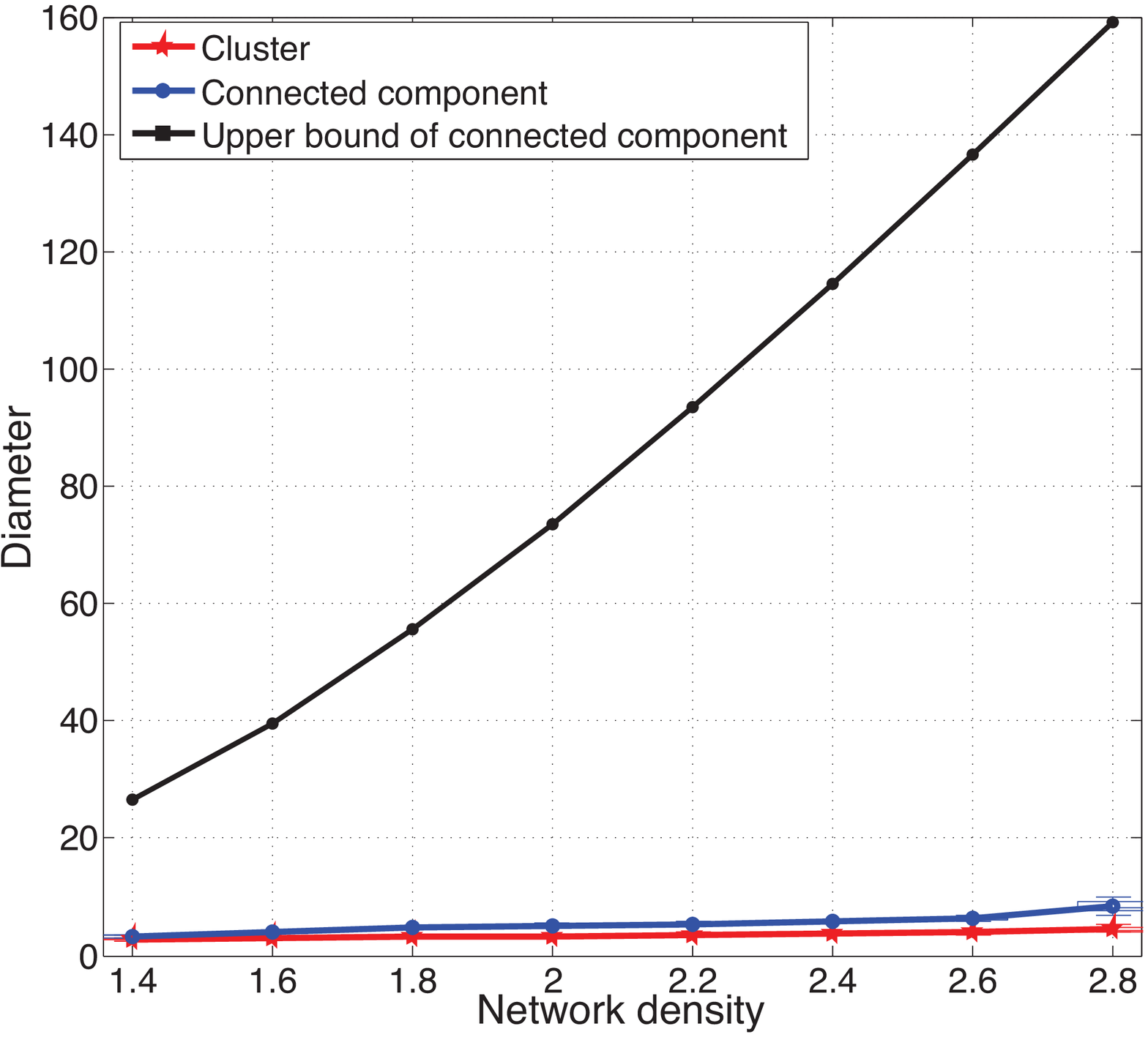}%
		\label{fig:xdiameter}%
		}
        \subfigure[]{%
		\includegraphics[height=1.6in]{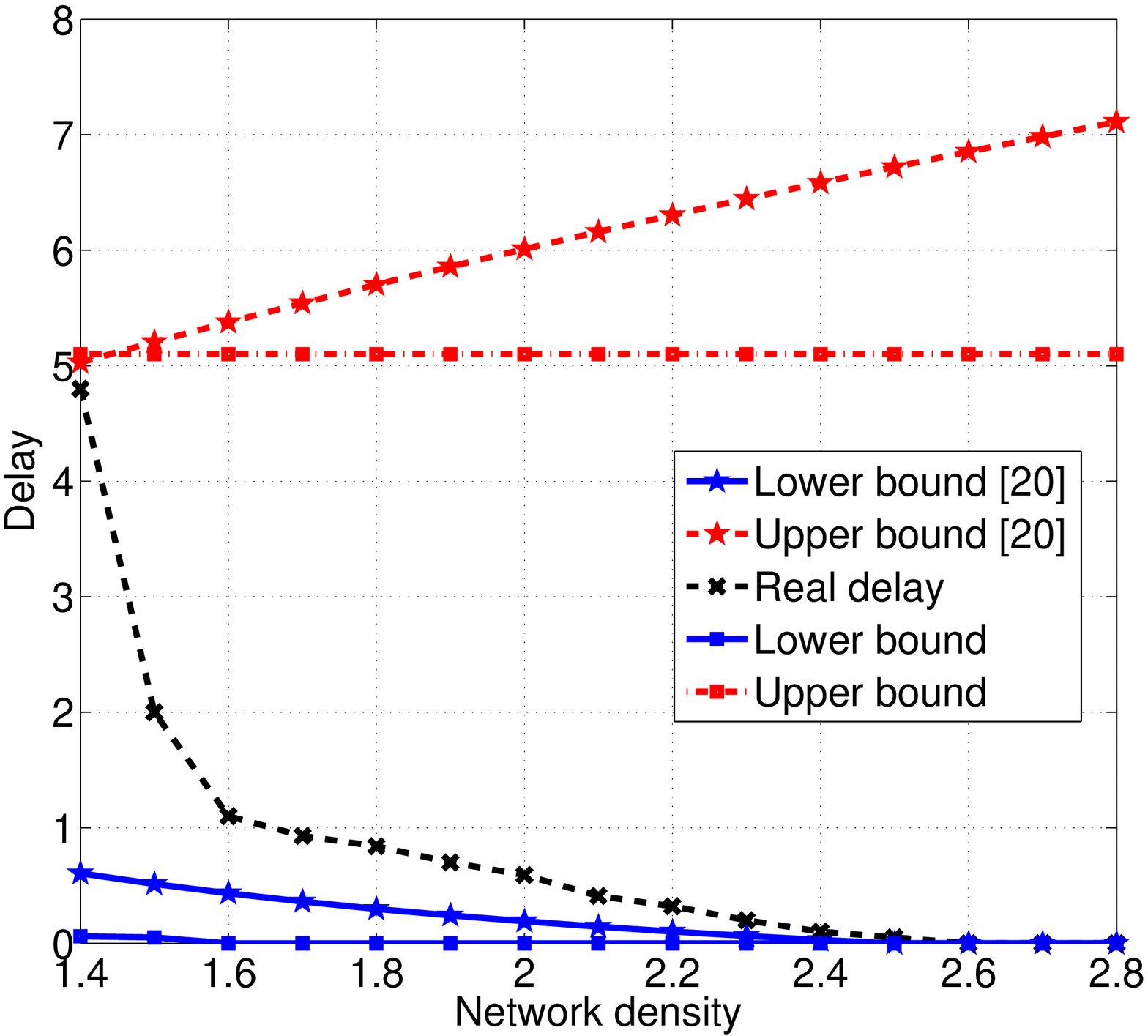}%
		\label{fig:delay}%
		}
		\subfigure[]{%
		\includegraphics[height=1.6in]{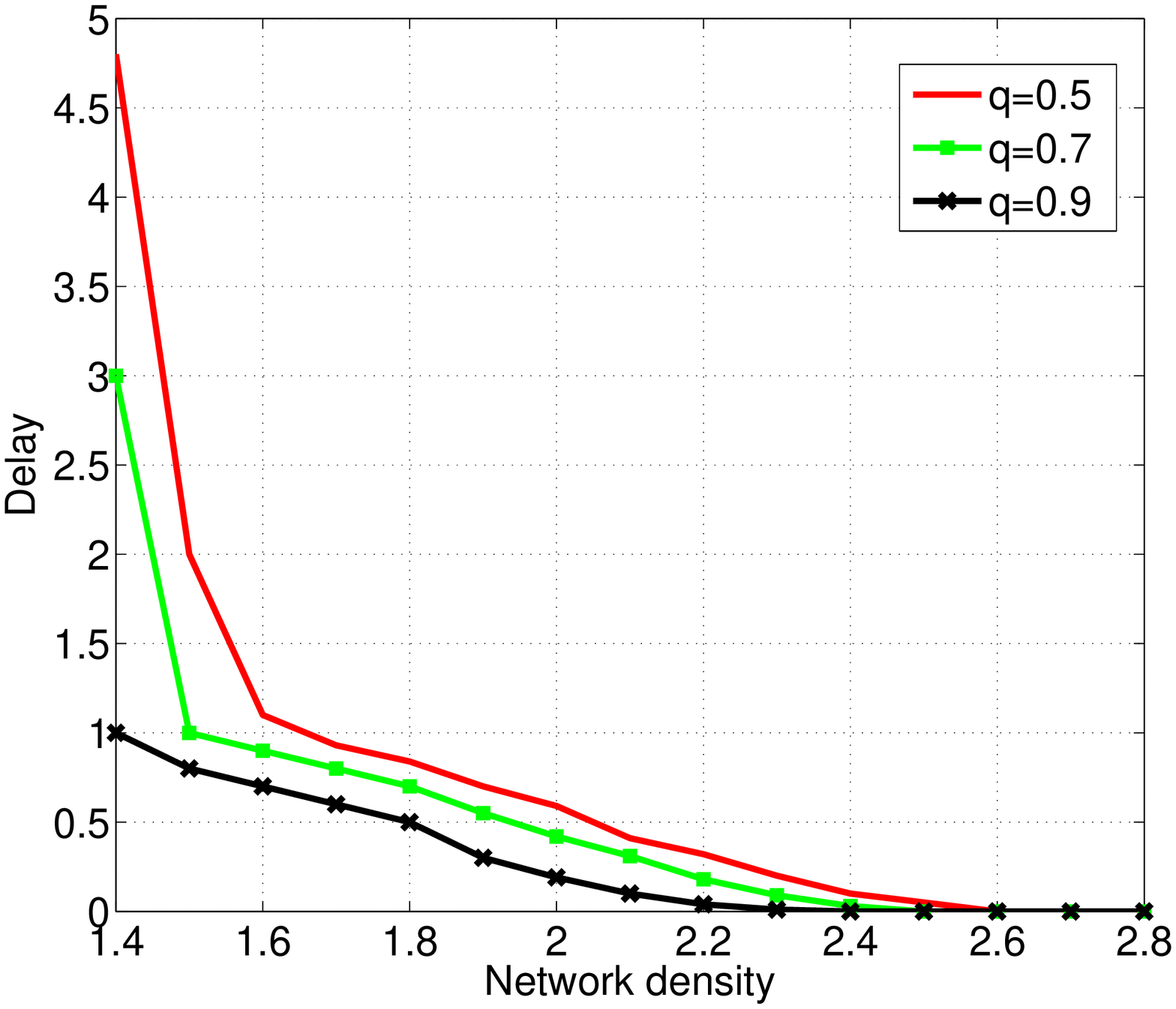}%
		\label{fig:connection}%
		}
\caption{(a) Cluster diameters and connected component diameters with various $\lambda$. (b) Upper bound, lower bound, and average value of $\gamma(\lambda)$. (c) Relation of relative network delay, energy harvesting rate, and node density.}
\end{figure*}

In the same figure, we further plot the upper and lower bounds of $\gamma(\lambda)$. 
From the figure, we see our upper bound is improved compared to that identified in~\cite{Wang11}. 
Our lower bound is slightly worse than that in~\cite{Wang11} because we use connected component to approximate cluster, and connected component is usually larger than cluster. 

To understand the relation of relative network delay, energy harvesting rate, and node density, we study $\gamma(\lambda)$ with various $q$'s, e.g., $q=0.5$, $q=0.7$, and $q=0.9$. We find $\gamma(\lambda)$ decreases when the energy harvesting rate $q$ increases. When the probability that a node harvests enough energy is $q=0.9$ (higher energy harvesting rate), the network is connected within the first time slot, \ie, the network is percolated for all $\lambda > \lambda_L = 1.4$.
\section{Conclusions}
\label{sec:conclusion}
In this paper, we study the fundamental limits of relative network delay $\gamma(\lambda)$ in an MEHWN. 
Modeling the energy availability on a node as an alternative renewal process, we prove that $\gamma(\lambda)$ is bounded by $\gamma(\lambda_L)$ and $\frac{1}{\left(  \sum_{n=2}^{\infty} \bar p_n \sum_{k=1}^{n-1} k\sum_{a=k}^{n-1} C_{n-1}^{a}  \left( \frac{1}{2} \right)^{n-1} \left( \frac{1}{k} \right)^{a-k} + 1 \right) r_0}$ where $\bar p_n$ is the upper bound of the probability that an arbitrary connected component's size is $n$. 

\begin{appendices}
\section{Proof of Upper Bound of $\gamma(\lambda)$}
Generating a sparse network $G'(\lambda_L, r_0, g)$ from the original network $G(\lambda, r_0, g)$ is shown Fig.~\ref{fig:transform} where $(\lambda_L / \lambda)$ is the probability of a node being kept from the original network. 
We see some nodes in $G(\lambda, r_0, g)$ do not exist in graph $G'(\lambda_L, r_0, g)$. 
According to the definition of $\lambda_L$, long-term connectivity is guaranteed in the network $G'(\lambda_L, r_0, g)$, i.e., there is a giant component $\mathcal{C}(G'(\lambda_L, r_0, g))$ containing most nodes from $G'(\lambda_L, r_0, g)$. 
Because $\lambda > \lambda_L$, we know $G(\lambda, r_0, g)$ is also percolated. 
\begin{figure}[h]
\centering
\includegraphics[width=3in]{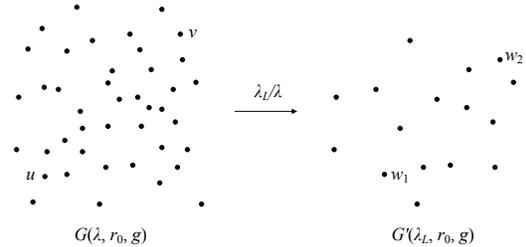}
\caption{Deriving a sparser network with node density $\lambda_L$ from a dense network with node density $\lambda$.}\label{fig:transform}
\end{figure}
For any node $u \in G(\lambda, r_0, g)$, we find a node $w$ from graph $\mathcal{C}(G'(\lambda_L, r_0, g))$ such that $d(u,w)$ is minimized.
\begin{lemma}
\label{distance}
Let $u \in G(\lambda, r_0, g)$, we have 
\[
w = \argmin \limits_{i \in \mathcal{C}(G'(\lambda_L, r_0, g))} \{ d(i, u) \},
\]
and $d(u,w) < \infty$.
\end{lemma}

\begin{proof}
The proof of Lemma~\ref{distance} can be found in~\cite{Yeh07}.
\end{proof}

We consider node $w \in G'(\lambda_L, r_0, g)$ the ``connection'' of node $u$. Because $w \in G(\lambda, r_0, g)$, and $d(u, w) < \infty$, we know the expected number of hops from $u$ to $w$ is finite from the following lemma.
\begin{lemma}
\label{NumOfHop}
If the distance is finite between two nodes in a percolated network, the expected number of hops between them is also finite.
\end{lemma}

\begin{proof}
See Proposition 4 in~\cite{Dousse04}.
\end{proof}

For any two nodes $u, v \in G(\lambda, r_0, g)$, we can find two ``connections'' $w_1, w_2 \in G'(\lambda_L, r_0, g)$. Note that it is possible that $u, v \in \mathcal{C}(G'(\lambda_L, r_0, g))$. In this case, we have $u=w_1$ and $v=w_2$.
\begin{definition}
\label{TwoNodes}
\begin{equation}
w_1 = \argmin \{ d(i, u) \},w_2 = \argmin \{ d(i, v) \}
\end{equation} 
\end{definition}
Due to Lemma~\ref{distance}, we have $d(w_1, u) < \infty$ and $d(w_2, v) < \infty$. Due to Lemma~\ref{NumOfHop}, we also have $E[N_\lambda \left( d(w_1, u) \right) ] < \infty $ and $E[N_\lambda \left( d(w_2, v) \right) ] < \infty $. 

Now we are ready to prove the Theorem~\ref{thm1}.

\begin{proof}
According to Definition~\ref{DelayDef}, $u$ and $v$ are two randomly selected nodes from $G(\lambda, r_0, g)$. It is possible that $u, v \in \mathcal{C}(G'(\lambda_L, r_0, g))$, with the probability of $(\lambda_L / \lambda)^2$. In this case, any path $\pi \in G'(\lambda_L, r_0, g)$ from $u$ to $v$ must exist in $G(\lambda, r_0, g)$, i.e., $\gamma(\lambda) \leq \gamma(\lambda_L)$. 

Now let's look at the case where $u, v \notin \mathcal{C}(G'(\lambda_L, r_0, g))$. According to Definition~\ref{TwoNodes}, we can identify two nodes $w_1, w_2 \in \mathcal{C}(G'(\lambda_L, r_0, g))$ such that $w_1$ and $w_2$ are the closest nodes to $u$ and $v$, respectively. Due to Lemmas~\ref{distance} and~\ref{NumOfHop}, we have $d(u, w_1) < \infty $, $d(v, w_2) < \infty $, $E[N_{\lambda} \left( d(u, w_1) \right) ] < \infty $ and $E[N_{\lambda} \left( d(v, w_2) \right) ] < \infty $. Therefore, $\gamma(\lambda)$ can be written as
\[
\begin{split}
\gamma(\lambda) & = \lim_{d(u,v) \to \infty } \frac{T_\lambda(u, v)}{d(u, v)} \\
                & \leq \lim_{d(u,v) \to \infty } \frac{T_\lambda(u, w_1) + T_\lambda(w_1, w_2) + T_\lambda(v, w_2) }{d(w_1, w_2) - d(u,w_1)-d(v,w_2)}.
\end{split}
\]

For a particular path $\pi_m$ from $u$ to $w_1$, we can calculate the delay of crossing it as 
\begin{equation*}
T_p(\pi_m) = \sum_{e \in \pi_m} {T(e)} = N_\lambda(d(u, w_1))E[T(e)].
\end{equation*}
On the other hand, $T_\lambda(u, w_1)$ denotes the first-passage time from $u$ to $w_1$, it is smaller than the delay of transmitting data along any path. Therefore, we have $$T_\lambda(u, w_1) \leq N_\lambda(d(u, w_1))E[T(e)],$$ and
\[
\begin{split}
&\gamma(\lambda) \leq \lim_{d(u,v) \to \infty } \\
& \frac{\left[N_\lambda(d(u, w_1)) + N_\lambda(d(v, w_2)) \right]E[T(e)] + T_\lambda(w_1, w_2)}{d(w_1, w_2) - d(u,w_1)-d(v,w_2)}.
\end{split}
\]
Since $N_\lambda(d(u, w_1)) < \infty$, $N_\lambda(d(v, w_2)) < \infty$, $d(u, w_1) < \infty$, $d(v, w_2) < \infty$, and $E[T(e)] < \infty$, we obtain
\begin{equation*}
\gamma(\lambda) \leq \lim_{d(w_1,w_2) \to \infty } \frac{T_\lambda(w_1, w_2)}{d(w_1, w_2)} 
\end{equation*}

As network $G'(\lambda_L, r_0, g)$ is derived from $G(\lambda, r_0, g)$ by randomly removing nodes, a path $\pi \in G'(\lambda_L, r_0, g)$ from $w_1$ to $w_2$ must exist in $G(\lambda, r_0, g)$. Therefore, 
\begin{equation}
\label{bound}
\gamma(\lambda) \leq \lim_{d(w_1,w_2) \to \infty } \frac{T_{\lambda_L}(w_1, w_2)}{d(w_1, w_2)} = \gamma({\lambda_L})
\end{equation}

Applying the same technique to the cases where only $u$ or $v$ in $ \mathcal{C}(G'(\lambda_L, r_0, g))$, we still have $\gamma(\lambda) \leq \gamma({\lambda_L})$.
\end{proof}

\end{appendices}
}

\bibliographystyle{IEEEtran}
\bibliography{energy}

\end{document}

%% file: new-commands.tex
\mathchardef\Gamma="0100 \mathchardef\Delta="0101
\mathchardef\Theta="0102 \mathchardef\Lambda="0103
\mathchardef\Xi="0104 \mathchardef\Pi="0105
\mathchardef\Sigma="0106 \mathchardef\Upsilon="0107
\mathchardef\Phi="0108 \mathchardef\Psi="0109
\mathchardef\Omega="010A

\newcommand{\outline}[1]{}

\usepackage{cite}
\usepackage{xspace}
\usepackage{url}
\usepackage{graphicx}
\usepackage{latexsym}
\usepackage{psfrag}
\usepackage[tight]{subfigure}
\usepackage{wrapfig}
\usepackage{comment}
\usepackage[ruled,linesnumbered,vlined]{algorithm2e}
\usepackage{alltt}
\usepackage{color}
\usepackage{setspace}
\usepackage{rotating}
\usepackage{enumerate}

\newcommand{\ie}{\emph{i.e.}\xspace}

\newtheorem{lemma}{Lemma}[section]

\newtheorem{theorem}{Theorem}[section]

\newcommand{\Comment}[1]{}